\documentclass[runningheads]{llncs}
\usepackage[T1]{fontenc}
\usepackage{algorithm,algpseudocode,graphicx,xcolor}

\newcommand{\LCP}{\ensuremath{\mathrm{LCP}}}
\newcommand{\LCS}{\ensuremath{\mathrm{LCS}}}
\newcommand{\ZFT}{\ensuremath{\mathrm{ZFT}}}
\newcommand{\expansion}[1]{\ensuremath{\langle #1 \rangle}}

\begin{document}

\title{Suffixient Sets}

\author{Lore Depuydt\inst{1}\orcidID{0000-0001-8517-0479} \and
Travis Gagie\inst{2}\orcidID{0000-0003-3689-327X} \and\\
Ben Langmead\inst{3}\orcidID{0000-0003-2437-1976} \and
Giovanni Manzini\inst{4}\orcidID{0000-0002-5047-0196} \and
Nicola Prezza\inst{5}\orcidID{0000-0003-3553-4953}}

\authorrunning{L. Depuydt et al.}

\institute{Department of Information Technology, Ghent University, Belgium \and
Faculty of Computer Science, Dalhousie University, Canada \and
Department of Computer Science, Johns Hopkins University, USA \and
Department of Computer Science, University of Pisa, Italy \and
Department of Environmental Sciences, Informatics and Statistics,\\
Ca' Foscari University of Venice, Italy}

\maketitle

\begin{abstract}
We define a suffixient set for a text $T [1..n]$ to be a set $S$ of positions between 1 and $n$ such that, for any edge descending from a node $u$ to a node $v$ in the suffix tree of $T$, there is an element $s \in S$ such that $u$'s path label is a suffix of $T [1..s - 1]$ and $T [s]$ is the first character of $(u, v)$'s edge label.  We first show there is a suffixient set of cardinality at most $2 \bar{r}$, where $\bar{r}$ is the number of runs in the Burrows-Wheeler Transform of the reverse of $T$.  We then show that, given a straight-line program for $T$ with $g$ rules, we can build an $O (\bar{r} + g)$-space index with which, given a pattern $P [1..m]$, we can find the maximal exact matches (MEMs) of $P$ with respect to $T$ in $O (m \log (\sigma) / \log n + d \log n)$ time, where $\sigma$ is the size of the alphabet and $d$ is the number of times we would fully or partially descend edges in the suffix tree of $T$ while finding those MEMs.

\keywords{Suffixient sets \and Maximal exact matches \and Suffix trees \and Burrows-Wheeler Transform.}
\end{abstract}

\section{Introduction}
\label{sec:introduction}

If we have the suffix tree of a text $T [1..n]$ then, given a pattern $P [1..m]$, we can compute in $O (m)$ time all the maximal exact matches (MEMs) of $P$ with respect to $T$ by starting at the root and repeatedly descending until we can descend no further and then following suffix links until we can descend again.  A MEM (sometimes also called a super MEM or SMEM) is a substring $P [i..j]$ such that either $i = 1$ or $P [i - 1..j]$ does not occur in $T$, and either $j = m$ or $P [i..j + 1]$ does not occur in $T$.  Figure~\ref{fig:descents} shows the 11 times we fully or partially descend 10 distinct edges while finding the MEMs when $m = 34$ and $n = 35$ and
\begin{eqnarray*}
P & = & \mathsf{1001001010010010100100101001010010} \\
T & = & \mathsf{0100101001001010010100100101001001\$} \,.
\end{eqnarray*}

\begin{figure}[t]
\begin{center}
\includegraphics[width=.88\textwidth]{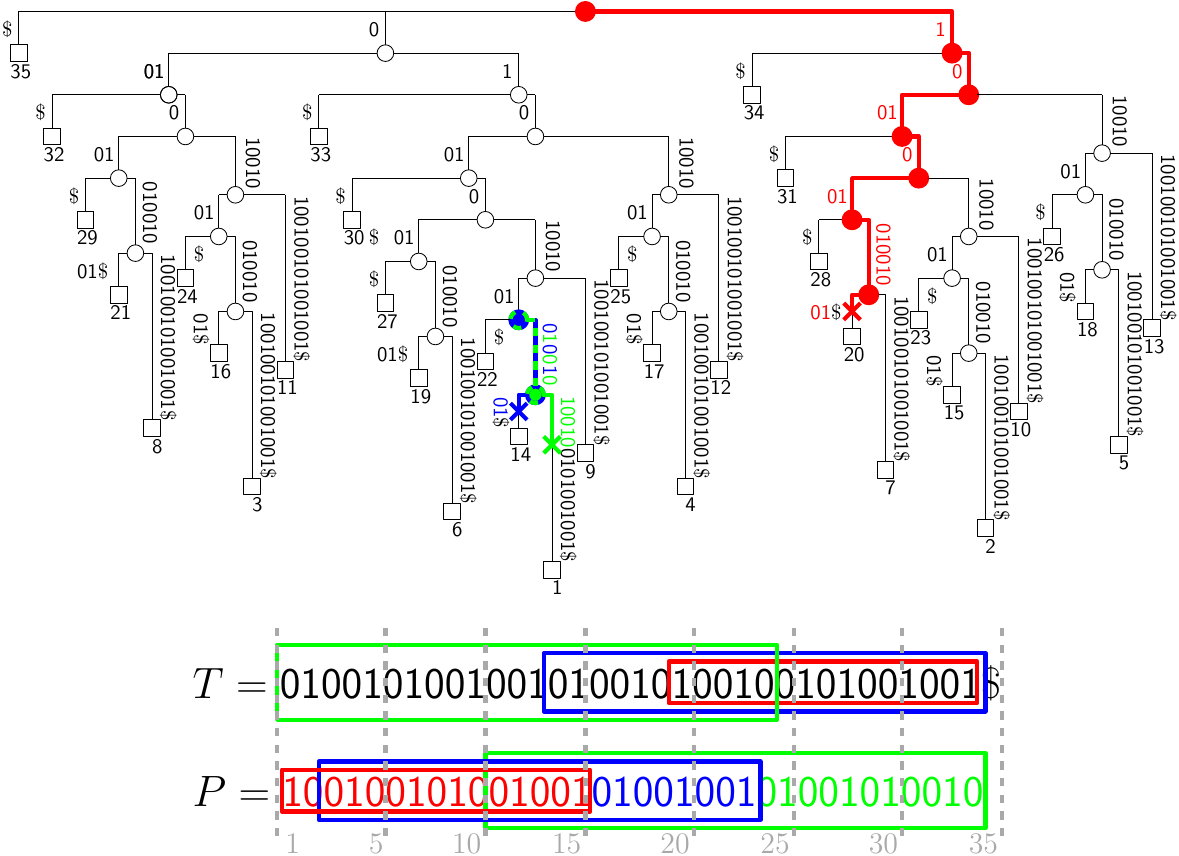}
\caption{The 11 times we fully or partially descend 10 distinct edges in the suffix tree of $T$ {\bf (above)} while finding the 3 MEMs for our example {\bf (below)}.  The MEMs are shown boxed in $P$ and $T$, with the characters' colours in $P$ also indicating which path we are following in the tree when we read them.  The characters in the box for a MEM that are a different colour from the box are the path label of the the node we reach by suffix links and descend from when finding the end of that MEM.  We descend the line alternating blue and green twice.}
\label{fig:descents}
\end{center}
\end{figure}

If we store the Karp-Rabin hashes of edges' labels and preprocess $P$ in $O (m \log (\sigma) / \log n)$ time such that we have constant-time access to the hashes of its substrings~\cite{Pre18}, then we can fully descend edges in constant time.  To also partially descend edges quickly, we can store
\begin{itemize}
\item for each edge, the starting position in $T$ of an occurrence of that edge's label,
\item a data structure that, given $i$ and $j$ and constant-time access to the hashes of $P$'s substrings, quickly returns the length $\LCP (P [i..m], T [j..n])$ of the longest common prefix of $P [i..m]$ and $T [j..n]$.
\end{itemize}
If the label of the next edge we should descend has length $\ell$ and its hash does not match the hash of the next $\ell$ characters $P [i..i + \ell - 1]$ of $P$, then $\LCP (P [i..m], T [j..n])$ tells us how far down the edge we can descend, where $j$ is the stored starting position of an occurrence of the edge's label.

Given a straight-line program (SLP) for $T$ with $g$ rules, we can turn it into an $O (g)$-space data structure that answers such LCP queries --- and also symmetric longest common suffix (LCS) queries --- in $O (\log n)$ time; see the appendix for details.  This gives us an augmented suffix tree that takes $O (n + g)$ space and finds all the MEMs of $P$ with respect to $T$ in $O (m \log (\sigma) / \log n + d)$ time plus $O (\log n)$ time per MEM, where $d$ is the number of times we fully or partially descend edges in the suffix tree (so $d = 11$ for our example).  There is a low probability of error due to the Karp-Rabin hashing.

MEM-finding has been important in bioinformatics at least since Li~\cite{Li13} introduced BWA-MEM, but linear-space and even entropy-compressed data structures are impractical when dealing with many massive but highly repetitive datasets such as pangenomes.  In this paper we show how to reduce our space bound to $O (\bar{r} + g)$, where $\bar{r}$ is the number of run in the Burrows-Wheeler Transform (BWT) of the reverse of $T$ (see~\cite{KK22} for discussion), while increasing the query time only to $O (m \log (\sigma) / n + d \log n)$ (still with a low probability of error due to the Karp-Rabin hashing).   The only indexes we know with comparable functionality and space bounds~\cite{GNP20,KK23,Nav23,ROLGB22} are significantly more complicated and either larger or slower in at least some cases.  Our compressed index is based on the new notion of suffixient sets, which may be of independent interest.

\section{Definitions and Size Bounds}
\label{sec:definitions_and_size_bounds}

We say a set of positions in $T$ is suffixient if the suffixes of the prefixes of $T$ ending at those positions are sufficient to cover the suffix tree of $T$ in a certain way:

\begin{definition}
\label{def:suffixient_set}
A {\em suffixient set} for a text $T [1..n]$ is a set $S$ of numbers between 1 and $n$ such that, for any edge descending from a node $u$ to a node $v$ in the suffix tree of $T$, there is an element $s \in S$ such that $u$'s path label is a suffix of $T [1..s - 1]$ and $T [s]$ is the first character of $(u, v)$'s edge label.
\end{definition}

\noindent An equivalent way to define suffixient sets is in terms of right-maximal substrings of $T$, which are substrings that occur in $T$ immediately followed by at least two distinct characters:

\begin{definition}
\label{def:right-maximal}
A {\em suffixient set} for a text $T [1..n]$ is a set $S$ of numbers between 1 and $n$ such that, for any right-maximal substring $\alpha$ of $T$ and any character $c$ that immediately follows an occurrence of $\alpha$ in $T$, there is an element $s \in S$ such that $\alpha\,c$ is a suffix of $T [1..s]$.
\end{definition}

\noindent The following lemma shows there is always a suffixient set for $T$ of cardinality at most $2 \bar{r}$.  We note that suffixient sets can be even smaller, however: for $T$ in our example $\bar{r} = 9$ but $\{14, 20, 33, 35\}$ is still suffixient, as illustrated in Figure~\ref{fig:cover}.

\begin{figure}[t]
\begin{center}
\includegraphics[width=.88\textwidth]{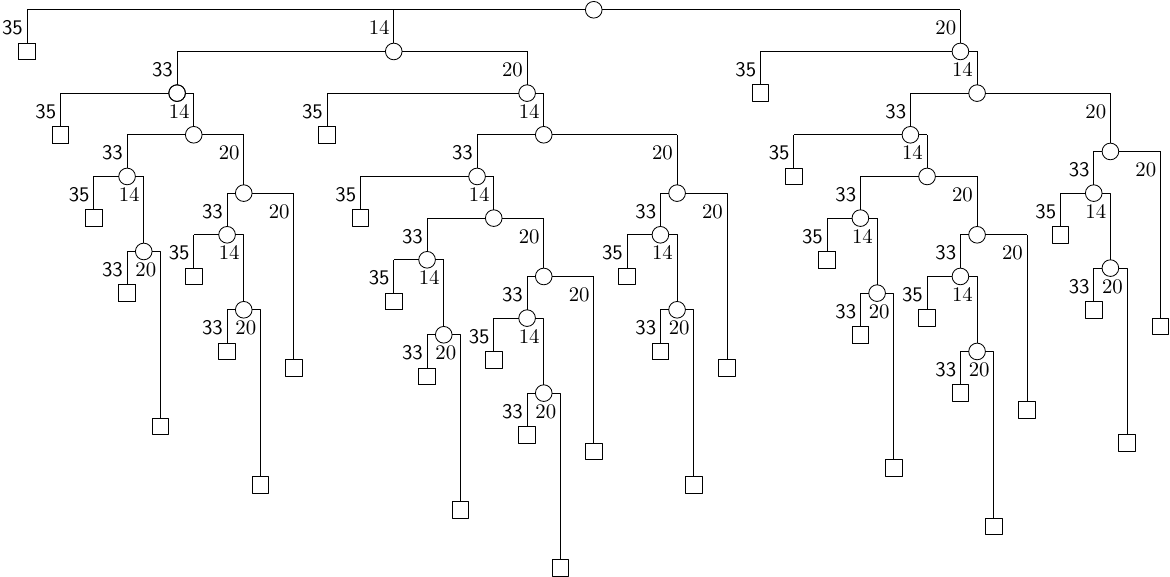}
\caption{The set $\{14, 20, 33, 35\}$ is suffixient for $T$ in our example.}
\label{fig:cover}
\end{center}
\end{figure}

\begin{lemma}
\label{lem:upper_bound}
The set of positions in $T$ of characters at the at most $2 \bar{r}$ run boundaries in the BWT of the reverse of $T$, is suffixient for $T$.
\end{lemma}

\begin{proof}
Consider an edge $(u, v)$ descending from a node $u$ to a node $v$ in the suffix tree of $T$.  Let $\alpha$ be $u$'s path label and $c$ be the first character in $(u, v)$'s edge label.  By the definition of the BWT, the occurrences of characters immediately following occurrences of $\alpha$ in $T$ are consecutive in the BWT of the reverse of $T$.  By the definition of suffix trees, $v$ must have a sibling, so not all the characters immediately following occurrences of $\alpha$ in $T$ are copies of $c$.  Therefore, for some $s$ such that $T [s]$ is at a run boundary in the BWT of the reverse of $T$, we have $T [s - |\alpha|..s] = \alpha\,c$.
\qed
\end{proof}

\noindent It is not difficult to show that any suffixient set for $T$ is also a string attractor~\cite{KP18} for $T$ --- that is, for any non-empty substring $T [i..j]$ of $T$, some occurrence of $T [i..j]$ contains $T [s]$ with $s \in S$ --- so, if we denote the cardinality of the smallest suffixient set for $T$ by $\chi$, we have $\gamma \leq \chi \leq 2 \bar{r}$, where $\gamma$ is the size of the smallest string attractor for $T$.  Due to space constraints, however, we leave the proof to Appendix~\ref{app:omitted_proof}.  In the full version of this paper we will combine this result with Kempa and Prezza's~\cite{KP18} results on string attractors to show there is an $O (\chi \log (n / \chi))$-space index with which we can find the MEMs of $P$ with respect to $T$ in $O (m \log (\sigma) / \log n + d \log n)$ time.

\begin{lemma}
\label{lem:attractor}
Any suffixient set $S$ for $T$ is a string attractor for $T$.
\end{lemma}

\noindent 

\section{Compressed Index}
\label{sec:compressed_index}

For clarity and without loss of generality we assume that all the characters in $P$ occur in $T$; otherwise, we can split $P$ into maximal substrings containing only character that do occur in $T$ and process those substrings independently.

Suppose we are given an SLP for $T$ with $g$ rules and we build a suffixient set $S$ for $T$ with cardinality at most $2 \bar{r}$.  The two components of our compressed index are the SLP-based LCP/LCS data structure mentioned in Section~\ref{sec:introduction} and described in Appendix~\ref{app:SLP}, and a z-fast trie~\cite{BBPV09} for the reversed prefixes of $T$ ending at positions in $S$.  The z-fast trie takes $O (\bar{r})$ space and, given $i$ and constant-time access to the hashes of $P$'s substrings, returns in $O (\log \min (m, n))$ time an element $\ZFT (P [1..i]) \in S$ such that $T [1..\ZFT (P [1..i])]$ has the longest common suffix with $P [1..i]$ of any prefix of $T$ ending at a position in $S$.

Algorithm~\ref{alg:MEMs} shows pseudocode for how we query our compressed index to find the MEMs of $P$ with respect to $T$.  Figures~\ref{fig:substrings} and~\ref{fig:trace} in Appendix~\ref{app:omitted_figures} show the prefixes of $T$ ending at positions in the suffixient set $\{14, 20, 33, 35\}$ in our example, the suffixes of $T$ that immediately follow them, the substrings of $P$ we consider while running Algorithm~\ref{alg:MEMs}, and a trace of how Algorithm~\ref{alg:MEMs} runs on our example.

\begin{algorithm}[t]
\caption{Pseudocode for finding the MEMs of $P [1..m]$ with respect to $T [1..n]$ using our compressed index.}
\label{alg:MEMs}
\begin{algorithmic}[1]
\State{$i \gets 1$}
\State{$\ell \gets 0$}
\While{$i \leq m$} \label{line:loop}
  \State{$j \gets \ZFT (P [1..i])$} \label{line:ZFT}
  \State{$b \gets \LCS (P [1..i], T [1..j])$} \label{line:LCS}
  \If{$b \leq \ell$} \label{line:MEM_check}
    \State{{\bf report} $(i - \ell, i - 1)$} \label{line:report}
  \EndIf
  \State{$f \gets \LCP (P [i + 1..m], T [j + 1..n])$} \label{line:LCP}
  \State{$i \gets i + f + 1$} \label{line:update_i}
  \State{$\ell \gets b + f$} \label{line:update_ell}
\EndWhile
\State{{\bf report} $(i - \ell, i - 1)$}
\end{algorithmic}
\end{algorithm}

Our algorithm is quite simple but we admit that it is not immediately obvious why it is correct, nor how to bound its running time.  Our arguments for both rests on the following technical lemma:

\begin{lemma}
\label{lem:invariant}
Whenever we reach the start of the while-loop in Line~\ref{line:loop} with $i \leq m$,
\begin{itemize}
\item we have reported all the MEMs that end strictly before $P [i - 1]$,
\item the longest common suffix of $P [1..i - 1]$ and any prefix of $T$ has length $\ell$,
\item the longest common suffix of $P [1..i]$ and any prefix of $T$ is the path label of some node $u$ in the suffix tree of $T$, followed by the first character in the label of the edge from $u$ to one of its children.
\end{itemize}
\end{lemma}

\begin{proof}
We proceed by induction on the number of times we have passed through the loop.  Clearly our inductive hypothesis --- the three points above --- is true when we first reach the start of the loop: $i = 1$, the longest common suffix of the empty prefix of $P$ and any prefix of $T$ has length $\ell = 0$, there are no MEMs in the empty prefix of $P$, the node $u$ is the root of the suffix tree and, since we can assume every character in $P$ occurs in $T$, there is an edge from $u$ to one of its children $v$ whose label starts with $P [i]$.

Assume our inductive hypothesis holds when we have passed through the loop $k \geq 0$ times.  By Definition~\ref{def:suffixient_set}, there is some element $s \in S$ such that $P [i - \ell..i]$ is a suffix of $T [1..s]$.  Therefore, after Lines~\ref{line:ZFT} and~\ref{line:LCS}, $T [1..j]$ has the longest common suffix with $P [1..i]$ of any prefix of $T$ and that common suffix has length $b$.  If and only if $b \leq \ell$ in Line~\ref{line:MEM_check} then the longest common suffix of $P [1..i]$ and any prefix of $T$ is no longer than the longest common suffix $P [i - \ell..i - 1]$ of $P [1..i - 1]$ and any prefix of $T$, so $P [i - \ell..i - 1]$ is a MEM and we correctly report it in Line~\ref{line:report}.

Whether $b \leq \ell$ or $b = \ell + 1$, the next MEM we should report starts at $P [i - b + 1]$.  After Line~\ref{line:LCP} we have $P [i - b + 1..i + f] = T [j - b + 1..j + f]$ but either $i + f + 1 = m + 1$ or $P [i + f + 1] \neq T [j + f + 1]$.  After we reset $i \leftarrow i + f + 1$ and $\ell \leftarrow b + f$ in Lines~\ref{line:update_i} and~\ref{line:update_ell}, respectively, at the end of the loop
\begin{itemize}
\item we have reported all the MEMs that end strictly before $P [i - 1]$,
\item the longest common suffix $P [i - \ell..i - 1] = T [j - b + 1..j + f]$ of $P [1..i - 1]$ and any prefix of $T$ has length $\ell$.
\end{itemize}

\noindent If $i = m + 1$ at the end of the loop then our inductive hypothesis is trivally true after we have passed through the loop $k + 1$ times, so assume $i \leq m$.  To show the final point of our inductive hypothesis holds, we consider two cases.  First, suppose $P [i - \ell..i]$ occurs in $T$; then since $P [i - \ell..i - 1] = T [j - b + 1..j + f]$ but $P [i] \neq T [j + f + 1]$, there are occurrences of $P [i - \ell..i - 1]$ in $T$ followed by different characters, so $P [i - \ell..i - 1]$ is the path label of some node $u$ in the suffix tree of $T$ and $P [i]$ is the first character in the label of the edge from $u$ to one of its children.

Now suppose $P [i - \ell..i]$ does not occur in $T$, meaning $P [i - \ell..i - 1]$ occurs in $T$ followed only by characters different than $P [i]$.  Consider the longest suffix $P [i - \ell'..i]$ of $P [i - \ell..i]$ that does occur in $T$, with $\ell' < \ell$.  Since $P [i - \ell'..i - 1]$ occurs in $T$ followed both by $P [i]$ and by another character, $P [i - \ell'..i - 1]$ is the path label of some node $u$ in the suffix tree of $T$ and $P [i]$ is the first character in the label of the edge from $u$ to one of its children.
\qed
\end{proof}

\noindent Since $i$ increases every time we pass through the loop, the second point of Lemma~\ref{lem:invariant} guarantees we report every MEM, and the third point guarantees that each time we pass through the loop corresponds to a different time we would full or partially descend an edge in the suffix tree of $T$ while finding the MEMs.  This gives us our result:

\begin{theorem}
\label{thm:main}
Given an SLP with $g$ rules for $T [1..n]$, we can build an $O (\bar{r} + g)$-space compressed index for $T$, where $\bar{r}$ is the number of runs in the BWT of the reverse of $T$, such that when given $P [1..m]$ we can find the MEMs of $P$ with respect to $T$ correctly with high probability and in $O (m \log \sigma + d \log n)$ time, where $\sigma$ is the size of the alphabet and $d$ is the number of times we would fully or partially descend edges in the suffix tree of $T$ while finding those MEMs.
\end{theorem}

\begin{credits}

\subsubsection{\ackname} 

This paper is dedicated to Margaret Gagie (1939--2023).  Many thanks to Adri\'an Goga for helpful discussions.
LD was funded by PhD Fellowship FR (1117322N), Research Foundation – Flanders (FWO).
TG was funded by NSERC Discovery Grant RGPIN-07185-2020.
BL was funded by NIH grant R01HG011392.
GM was funded by the Italian Ministry of Health, POS 2014--2020, project ID T4-AN-07, CUP I53C22001300001; INdAM-GNCS Project CUP E53C23001670001; and PNRR ECS00000017 Tuscany Health Ecosystem, Spoke~6 CUP I53C22000780001.
NP was funded by the European Union (ERC, REGINDEX, 101039208).  Views and opinions expressed in this paper are, however, those of the authors only and do not necessarily reflect those of the European Union or the European Research Council.  Neither the European Union nor the granting authority can be held responsible for them.

\subsubsection{\discintname}

The authors have no competing interests to declare that are relevant to the content of this article.

\end{credits}

\appendix

\section{SLP-based LCP/LCS Data Structure}
\label{app:SLP}

Ganardi, Je\.z and Lohrey~\cite{GJL21} showed how, given an SLP for $T [1..n]$ with $g$ rules, we can build another SLP for $T$ with $O (g)$ rules and height $O (\log n)$, so assume without loss of generality that we are given an SLP with height $O (\log n)$.  For each symbol $X$ in the SLP, we store with $X$ the length $|\expansion{X}|$ and the Karp-Rabin hash $h (\expansion{X})$ of $X$'s expansion $\expansion{X}$.  With low probability the hashes of distinct substrings of $P$ and $T$ collide and cause errors.

We find $\LCP (P [i..m], T [j..n])$ recursively, starting with intervals $[i..m]$ and $[j..j + m - i]$ at the root of the parse tree of $T$.  (Finding $\LCS (P [1..i], T [1..j])$ is symmetric, so we omit its description.) Suppose that at some point we have arrived with intervals $[i'..i' + \ell - 1]$ and $[j'..j' + \ell - 1]$ at a symbol $X$, trying to find $\LCP \left( \rule{0ex}{2ex} P [i'..i' + \ell - 1], \expansion{X} [j'..j' + \ell - 1] \right)$.  If $X$ is a terminal then this takes constant time.  If $\ell = |\expansion{X}|$ and $h (P [i'..i' + \ell - 1]) = h (\expansion{X})$ and there is no collision then
\[\LCP \left( \rule{0ex}{2ex} P [i'..i' + \ell - 1], \expansion{X} [j'..j' + \ell - 1] \right)
= \ell\,.\]

\noindent Otherwise, suppose $X \rightarrow Y\,Z$.  If $\expansion{X} [j'..j' + \ell - 1]$ is completely contained in $\expansion{Y}$ then we recurse on $Y$ with the same intervals $[i'..i' + \ell - 1]$ and $[j'..j' + \ell - 1]$ to find
\begin{eqnarray*}
\lefteqn{\LCP \left( \rule{0ex}{2ex} P [i'..i' + \ell - 1], \expansion{X} [j'..j' + \ell - 1] \right)} \\
& = & \LCP \left( \rule{0ex}{2ex} P [i'..i' + \ell - 1], \expansion{Y} [j'..j' + \ell - 1] \right)\,.
\end{eqnarray*}
If $\expansion{X} [j ' - \ell + 1..j']$ is completely contained in $\expansion{Z}$ then we recurse on $Z$ with intervals $[i'..i' + \ell - 1]$ and $[j' - |\expansion{Y}|..j' + \ell - 1 - |\expansion{Y}|]$ to find
\begin{eqnarray*}
\lefteqn{\LCP \left( \rule{0ex}{2ex} P [i'..i' + \ell - 1], \expansion{X} [j'..j' + \ell - 1] \right)}\\
& = & \LCP \left( \rule{0ex}{2ex} P [i'..i' + \ell - 1], \expansion{Z} [j' - |\expansion{Y}|..j' + \ell - 1 - |\expansion{Y}|] \right)\,.
\end{eqnarray*}

\noindent If $\expansion{X} [j'..j' + \ell - 1]$ overlaps both $\expansion{Y}$ and $\expansion{Z}$ with $\ell'$ characters in $\expansion{Y}$ and $\ell - \ell'$ characters in $\expansion{Z}$ then we first recurse on $Y$ with intervals $[i'..i + \ell' - 1]$ and $[j'..|\expansion{Y}|]$ to find $\LCP \left( \rule{0ex}{2ex} P [i'..i' + \ell' - 1], \expansion{Y} [j'..|\expansion{Y}|] \right)$.  If
\[\LCP \left( \rule{0ex}{2ex} P [i'..i' + \ell' - 1], \expansion{Y} [j'..|\expansion{Y}|] \right)
< \ell'\]
then
\begin{eqnarray*}
\lefteqn{\LCP \left( \rule{0ex}{2ex} P [i'..i' + \ell - 1], \expansion{X} [j'..j' + \ell - 1] \right)} \\
& = & \LCP \left( \rule{0ex}{2ex} P [i'..i' + \ell' - 1], \expansion{Y} [j'..|\expansion{Y}|] \right)\,.
\end{eqnarray*}
Otherwise,
\begin{eqnarray*}
\lefteqn{\LCP \left( \rule{0ex}{2ex} P [i'..i' + \ell - 1], \expansion{X} [j'..j' + \ell - 1] \right)} \\
& = & \LCP \left( \rule{0ex}{2ex} P [i' + \ell'..i' + \ell - 1], \expansion{X} [1..\ell - \ell'] \right) + \ell'
\end{eqnarray*}
and we compute $\LCP \left( \rule{0ex}{2ex} P [i' + \ell'..i' + \ell - 1], \expansion{X} [1..\ell - \ell'] \right)$ by recursing on $Z$ with intervals $[i' + \ell'..i + \ell - 1]$ and $[1..\ell - \ell']$.

To see why this recursion takes $O (\log n)$ time, consider it as a binary tree.  Let $v$ be a leaf of that tree and let $u$ be its parent.  The expansion of the symbol corresponding to $v$ is completely contained in $T \left[ \rule{0ex}{2ex} j..j + \LCP (P [i..m], T [j..n]) - 1 \right]$, but the expansion of the symbol $X_u$ corresponding to $u$ is not.  This means $X_u$ is either on the path from the root of the parse tree to the $j$th leaf,  or on the path from the root to the $\left( \rule{0ex}{2ex} j + \LCP (P [i..m], T [j..n]) - 1 \right)$st leaf.  Since those paths have length $O (\log n)$, the recursion takes $O (\log n)$ time.

\section{Proof of Lemma~\ref{lem:attractor}}
\label{app:omitted_proof}

\begin{proof}
Consider a non-empty substring $T [i..j]$ of $T$, let $v$ be the locus of $T [i..j]$ in the suffix tree of $T$ (that is, the highest node whose path label is prefixed by $T [i..j]$), let $u$ be $v$'s parent, let $\alpha$ be $u$'s path label, and let $c$ be the first character of $(u, v)$'s label.  Then $\alpha\,c$ is a prefix of $T [i..j]$ and, since $(u, v)$ is a single edge, any occurrence of $\alpha\,c$ in $T$ is contained in an occurrence of $T [i..j]$.  By Definition~\ref{def:suffixient_set}, there is some $s \in S$ such that $\alpha\,c$ is a suffix of $T [1..s]$, so $T [s]$ is contained in an occurrence of $\alpha\,c$ and thus contained in an occurrence of $T [i..j]$.
\qed
\end{proof}

\pagebreak

\section{Omitted Figures}
\label{app:omitted_figures}

\begin{figure}[h!]
\begin{center}
\resizebox{\textwidth}{!}
{\begin{tabular}{rr|ll}
$T [1..35] =$  & \sf \textcolor{blue}{0100101001001010010100100101001001\$} &                                              &                \\[2ex]
$P [3..22] =$  &                  \sf \textcolor{red}{01001010010010100100} & \sf \textcolor{red}{1}                       & $= P [23]$     \\
$T [1..33] =$  &    \sf \textcolor{blue}{010010100100101001010010010100100} & \sf \textcolor{blue}{1\$}                    & $= T [34..35]$ \\[2ex]
$P [3..16] =$  &                        \sf \textcolor{red}{01001010010010} & \sf \textcolor{red}{10010}                   & $= P [17..21]$ \\
$P [11..24] =$ &                        \sf \textcolor{red}{01001010010010} & \sf \textcolor{red}{1001010010}              & $= P [25..34]$ \\                
$T [1..14] =$  &                       \sf \textcolor{blue}{01001010010010} & \sf \textcolor{blue}{10010100100101001001\$} & $= T [15..35]$ \\[2ex]
$P [1] =$      &                                     \sf \textcolor{red}{1} & \sf \textcolor{red}{00100101001001}          & $= P [2..15]$  \\
$T [1..20] =$  &                 \sf \textcolor{blue}{01001010010010100101} & \sf \textcolor{blue}{00100101001001\$}       & $= T [21..35]$
\end{tabular}}
\caption{The prefixes of $T$ ending at positions in the suffixient set $\{14, 20, 33, 35\}$ in our example {\bf (left, in blue)}, the suffixes of $T$ that immediately follow them {\bf (right, in blue)}, the longest common prefixes of those prefixes of $T$ with the prefixes of $P$ we consider with Algorithm~\ref{alg:MEMs} {\bf (left, in red)}, and the longest common suffixes of those suffixes of $T$ with the remaining suffixes of $P$ {\bf (right, in red)}.}
\label{fig:substrings}
\end{center}
\end{figure}

\begin{figure}[h!]
\begin{center}
\begin{tabular}{c@{\hspace{8ex}}c}
\begin{tabular}{rl}
 1 & $i \leftarrow 1$ \\
 2 & $\ell \leftarrow 0$ \\[2ex]
 3 & $1 \leq 34$ \\
 4 & $j \leftarrow \ZFT (P [1]) = 20$ \\
 5 & $b \leftarrow \LCS (P [1], T [1..20]) = 1$ \\
 6 & $1 \not \leq 0$ \\
 9 & $f \leftarrow \LCP (P [2..34], T [21..35]) = 14$ \\
10 & $i \leftarrow 16$ \\
11 & $\ell \leftarrow 15$ \\[2ex]
 3 & $16 \leq 34$ \\
 4 & $j \leftarrow \ZFT (P [1..16]) = 14$ \\
 5 & $b \leftarrow \LCS (P [1..16], T [1..14]) = 14$ \\
 6 & $14 \leq 15$ \\
 7 & {\bf report} $(1, 15)$ \\
 9 & $f \leftarrow \LCP (P [17..34], T [15..35]) = 5$ \\
10 & $i \leftarrow 22$ \\
11 & $\ell \leftarrow 19$ \\[2ex]
\end{tabular}
&
\begin{tabular}{rl}
 3 & $22 \leq 34$ \\
 4 & $j \leftarrow \ZFT (P [1..22]) = 33$ \\
 5 & $b \leftarrow \LCS (P [1..22], T [1..33]) = 20$ \\
 6 & $20 \not \leq 19$ \\
 9 & $f \leftarrow \LCP (P [23..34], T [34..35]) = 1$ \\
10 & $i \leftarrow 24$ \\
11 & $\ell \leftarrow 21$ \\[2ex]
 3 & $24 \leq 34$ \\
 4 & $j \leftarrow \ZFT (P [1..24]) = 14$ \\
 5 & $b \leftarrow \LCS (P [1..24], T [1..14]) = 14$ \\
 6 & $14 \leq 21$ \\
 7 & {\bf report} $(3, 23)$ \\
 9 & $f \leftarrow \LCP (P [25..34], T [15..35]) = 10$ \\
10 & $i \leftarrow 35$ \\
11 & $\ell \leftarrow 24$ \\[2ex]
 3 & $35 \not \leq 34$ \\
13 & {\bf report} $(11, 34)$
\end{tabular}
\end{tabular}
\caption{A trace of how Algorithm~\ref{alg:MEMs} runs on our example.}
\label{fig:trace}
\end{center}
\end{figure}

\end{document}